\documentclass[letterpaper, 10 pt, conference]{ieeeconf}  

\IEEEoverridecommandlockouts                              

\usepackage{graphicx}
\usepackage{subcaption}
\usepackage{comment}
\usepackage{epsfig} 
\usepackage{times} 
\usepackage{amsmath} 
\usepackage{amssymb}  
\usepackage{xcolor}
\usepackage{flushend}
\newtheorem{mydef}{\bf Definition}

\newtheorem{myprob}{\bf Problem}

\newtheorem{myprop}{\bf Proposition}

\newcounter{WSQcomment}

\title{\LARGE \bf 
SPARC: Prediction-Based Safe Control for Coupled Controllable and Uncontrollable Agents with Conformal Predictions}

\author{Shuqi Wang, Siqi Wang, Shaoyuan Li and Xiang Yin
\thanks{This work was supported by the National Natural Science Foundation of China (62173226,62061136004).}
\thanks{Shuqi Wang, Siqi Wang, Shaoyuan Li and Xiang Yin are with Department of Automation and Key Laboratory of System Control and Information Processing, Shanghai Jiao Tong University, Shanghai 200240, China. e-mail: \{wangshuqi, sq\_wang, syli, yinxiang\}@sjtu.edu.cn.}
}

\begin{document}

\maketitle
\thispagestyle{empty}
\pagestyle{empty}

\begin{abstract}
We investigate the problem of safe control synthesis for autonomous systems operating in environments with uncontrollable agents whose dynamics are unknown but coupled with those of the controlled system. 
This scenario naturally arises in various applications, such as autonomous driving and human-robot collaboration, where the behavior of uncontrollable agents, such as pedestrians or human operators, cannot be directly controlled but is influenced by the actions or the states of the autonomous vehicle or robot. 
In this paper, we present SPARC (\textsc{S}afe \textsc{P}rediction-Based \textsc{R}obust \textsc{C}ontroller for Coupled \textsc{A}gents), a new framework designed to ensure safe control in the presence of  uncontrollable agents that are coupled with the controlled system.  
Particularly, we introduce a joint distribution-based approach to account for the coupled dynamics of the controlled system and uncontrollable agents. Then, SPARC leverages conformal prediction to quantify the uncertainty in data-driven prediction of the agent behavior. By integrating the control barrier function (CBF), SPARC provides provable safety guarantees at a high confidence level.  We illustrate our framework with a case study of an autonomous driving scenario with walking pedestrians. 

\end{abstract}

\section{INTRODUCTION}

Safety is one of the major concerns in autonomous systems such as self-driving cars or industrial robots. These systems must operate in dynamic environments while avoiding unsafe behaviors. For instance, in autonomous driving, the ego vehicle must continuously monitor its surroundings to ensure the safety of itself, pedestrians, and other vehicles.
To achieve safe control in real-world scenarios, it is essential to account for the presence of uncontrollable agents in the environment. 

There are two major challenges in achieving safe control in such scenarios.
First, the dynamics of the controlled system and the uncontrollable agents are  \emph{coupled} in general. In other words, the behavior of the uncontrollable agent depends on the state of the controlled system. For example, in autonomous driving, a pedestrian's movements may change in response to the position and speed of a nearby vehicle. 
Second, the intentions of such agents are often \emph{unknown a priori} and difficult to model precisely. 
As a result, data-driven approaches have emerged as a powerful tool for handling the intentions of uncontrollable agents. By leveraging large datasets collected from real-world interactions, one can develop highly accurate predictors to forecast the trajectories of uncontrollable agents for control purposes. For instance, as shown in Figure~\ref{fig:instance}, in autonomous driving within pedestrian-rich environments, there are extensive datasets that capture how pedestrians react to vehicles under various conditions. 
However, these predictors fail to measure uncertainty in a formal and systematic manner, rendering the prediction-based control susceptible to unsafe outcomes.


\begin{figure}[tp]   
    \centering
    \begin{subfigure}[b]{0.22\textwidth}
        \centering
        \includegraphics[width=0.8\textwidth]{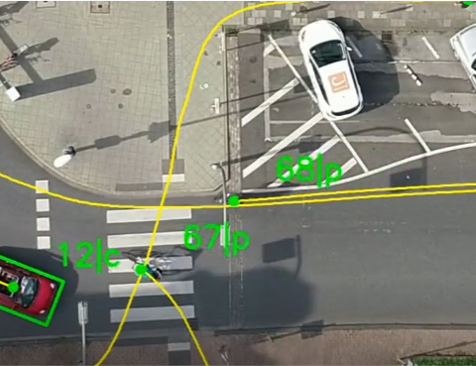}  
        \caption{InD Dataset\cite{inD}}  
        \label{fig:inD}
    \end{subfigure} 
    \begin{subfigure}[b]{0.22\textwidth}
        \centering
        \includegraphics[width=0.8\textwidth]{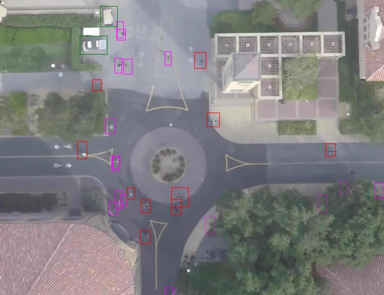}  
        \caption{Stanford Drone Dataset\cite{SDD}} 
        \label{fig:SDD}
    \end{subfigure}
    \caption{Instances of coupled dynamics in real-world scenarios can be observed across various datasets. For example, the interactions between traffic participants often exhibit mutual influence. In Fig.\ref{fig:inD}, vehicles may decelerate as they encounter bicycles at crosswalks. At the same time bicycles may alter their trajectories to avoid collisions with vehicles, resulting in curved paths.}
    \label{fig:instance}
\end{figure}

In this paper, we propose a new framework, called SPARC (\textbf{S}afe \textbf{P}rediction-Based \textbf{R}obust \textbf{C}ontroller for Coupled \textbf{A}gents), for ensuring safe control in the presence of uncontrollable agents. Specifically, we assume the dynamics of the uncontrollable agent are unknown and coupled with the dynamics of the controlled system. To address this challenge, we leverage the theory of conformal prediction to quantify uncertainty in data-driven predictions of the uncontrollable agent's intentions. Unlike previous approaches, which assume that the dynamics of the uncontrollable agent are independent of the controlled system, we develop a novel method to sample data and construct conformal regions based on the \emph{joint distribution} over an augmented state space. This approach allows us to capture the inter-dependencies between the two agents and improve prediction accuracy. Furthermore, by incorporating the control barrier function (CBF) techniques, we can provide provable safety guarantees with a confidence bound, ensuring robust system performance.
Our framework is particularly well-suited for applications such as autonomous driving in urban environments or human-robot collaboration, where uncontrollable agents operate based on their own underlying hidden decision-making processes that are influenced by the behavior of the controlled systems.

\begin{figure*}[t]   
    \centering \includegraphics[width=0.75\textwidth]{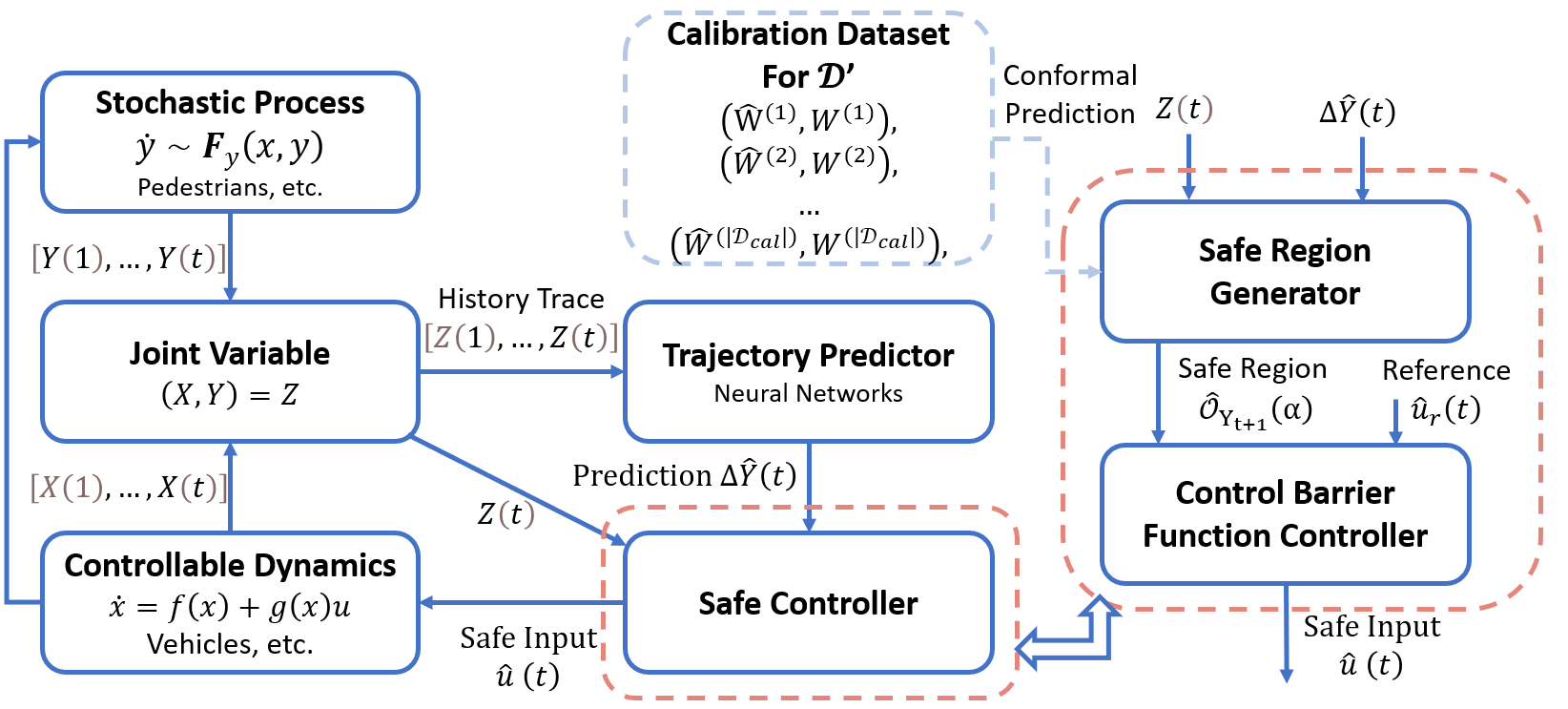} 
        \caption{SPARC (Our Framework)} 
        \label{fig:frame_now}
    \label{fig:frame}
\end{figure*}

\section{Related Works}

\textbf{Safety-Critical Control using CBFs. }
Control Barrier Functions (CBFs) are a widely used approach for ensuring safety in control systems \cite{yang2022differentiable,xiao2021high,xiao2021adaptive,xiaosafe,ames2019control}.
First introduced by Ames et al. \cite{ames2016control}, the key idea behind CBFs is to encode safety constraints into a   function that guarantees the system remains within a safe set during operation. However, a significant challenge of using CBFs is their reliance on an accurate dynamic model of the system. When the system includes an uncontrollable component with unknown dynamics, applying traditional CBFs can be problematic. In such cases, robust CBFs \cite{jankovic2018robust,buch2021robust} can be used to account for the worst-case scenario, which, however, often results in overly conservative control strategies.

\textbf{Data-Driven Safety Control Synthesis. }
When the system dynamic is \emph{unknown a priori}, many recent works have studied how to synthesize  safe controllers using data; see, e.g., \cite{nejati2023data,zhong2022synthesizing,ajeleye2024data,mitsioni2023safe,salamati2024data,chen2023data,yin2024formal}.
For instance, in \cite{salamati2024data}, the authors convert the safe control synthesis to a convex program by using barrier certificates. 
The scenario approach is used to provide a probabilistic safety guarantee based on sampled data. However, most of these works cannot handle the presence of uncontrollable agents whose dynamics are coupled with the control systems.

\textbf{Conformal Prediction in Control Synthesis. }
Conformal prediction is an emerging technique in machine learning that provides a principle way to quantify the uncertainty in model predictions \cite{shafer2008tutorial}. It works by generating prediction sets or confidence regions that, with a specified level of certainty, are guaranteed to contain the true outcome. This technique is model-free, meaning it can be applied to any underlying predictive model, and it offers formal statistical guarantees on the accuracy of its predictions.
Recently, conformal prediction has been extensively applied to control synthesis problems; see, e.g., \cite{lindemann2023safe,yang2023safe,dixit2023adaptive,yu2023,stamouli2024recursively,sun2024conformal}. 
Our work is closely related to \cite{yu2023} and \cite{lindemann2023safe}, which address the problem of safe planning in unknown dynamic environments. In these works, the authors also employ conformal prediction to estimate the confidence region of the uncontrollable agent. However, a key distinction is that these approaches assume the behavior of uncontrollable agents is independent of the controlled system, an assumption that may not hold in many real-world scenarios. In contrast, our framework considers a coupled setting, 
where the behavior of uncontrollable agents is influenced by the state of the controlled system.

\section{Problem Formulation}

\emph{Notations: }We denote $\mathbb{R}, \mathbb{N}$, and $\mathbb{R}^n$ as the set of real numbers, natural numbers, and real vectors, respectively. 
Let $\beta: \mathbb{R} \rightarrow \mathbb{R}$ denote an extended class $\mathcal{K}_{\infty}$ function, i.e., a strictly increasing function with $\beta(0)=0$. 
For a vector $v \in \mathbb{R}^n$, let $|v|$, $\|v\|$, $\|v\|_\infty$ denote its $\ell_1$-norm, Euclidean norm and $\ell_\infty$-norm. 
For a finite set $D$, let $|D|$ denote its cardinality.

\subsection{System Models}

\textbf{Controlled System:} 
The system under control is modeled as a nonlinear control-affine system of form
\begin{equation}\label{eq:control_dynamic}
\dot{x}(t) = f(x(t)) + g(x(t)) u(t) =: F(x(t), u(t)),
\end{equation}
where $x(t) \in \mathbb{R}^{n}$ represents the state of the system at time $t$, 
and $u(t) \in \mathcal{U}$ denotes the control input  with $\mathcal{U} \subseteq \mathbb{R}^c$ being the set of admissible control inputs.
The functions $f: \mathbb{R}^n \rightarrow \mathbb{R}^{n}$ and $g: \mathbb{R}^n \rightarrow \mathbb{R}^{n \times c}$ represent the internal system dynamics and the input influence, respectively, and both are assumed to be locally Lipschitz continuous.

\textbf{Uncontrollable Agent:}
We consider a scenario where the controlled system is operating in an environment with a single uncontrollable agent for the sake of simplicity. And it can be easily extended to $N$ uncontrollable agents in principle by concatenating their states and considering them as a whole.
Specifically, we assume that the behavior of the uncontrollable agent is coupled with the controllable agent, i.e., the evolution of the dynamics of uncontrollable agents is conditioned on the state of the controllable agent $x(t)$. 
For instance, in autonomous driving scenarios involving moving pedestrians, their behavior is strongly influenced by the car's position and velocity. Specifically, pedestrians tend to move away more quickly as the car gets closer, and they adopt more cautious walking strategies when the car approaches at higher speeds.
The uncontrollable agent dynamics can be modeled as a 
stochastic process
\begin{equation}
    \dot{y}(t) \sim \mathbb{F}_y(x(t),y(t)), 
    \label{eq:uncontrol_dynamic}
\end{equation}
where $y(t) \in \mathbb{R}^m$ is the state of the uncontrollable agent, and $\mathbb{F}_y$ is a distribution parameterized by the states of the controlled and uncontrollable systems, which is not known a priori.  


We frame safety as the 0-upper level set of a Lipschitz continuous function $h_{xy}(x,y)$, i.e., 
$\mathcal{S}=\{(x,y)\mid h_{xy}(x,y)\geq 0\}\subseteq  \mathbb{R}^{n}\times \mathbb{R}^{m}$. 
And our objective is to synthesize a safety filter for reference control at each time instant such that the controlled system can ensure safety with respect to the uncontrollable agent, with a certain level of confidence.

\begin{myprob}
    Given $h_{xy}$ defined with a safety set $\mathcal{S}$, a reference control $u_r$ at time $t$, we wish to obtain a safe control input $u^*$, that minimally interferes with the reference control while assuring safety. 
    Formally, we aim to solve the following problem
\begin{equation}
\begin{aligned}
     u^* &= \arg\min_{u \in \mathcal{U}}|| u - u_r ||^2  \\
     \text{s.t.} \quad 
     & \text{Prob}[h_{xy}(x(t'),y(t')) \geq 0] \geq 1 - \alpha , \; \forall t'> t, \\
     & x(t') \text{ satisfies the dynamics of } \eqref{eq:control_dynamic}, \\
     & y(t') \text{ satisfies the dynamics of } \eqref{eq:uncontrol_dynamic}.
\end{aligned}
\label{eq:optimal_goal}
\end{equation}
\end{myprob}\medskip

Since the behavior of the uncontrollable agent is complete unknown, we will train a neural network to model its behavior and provide confidence level guarantee for safety of form: ``\emph{the probability that the system becomes unsafe with the filtered input is lower than a threshold
$\alpha$}".

\section{ \textbf{S}afe \textbf{P}rediction-Based \textbf{R}obust \textbf{C}ontroller for Coupled  \textbf{A}gents   (SPARC)}

\subsection{Overall Control Framework}
Before we formally present the implementation details of the proposed controller, we provide a high level overview of the entire SPARC framework. 
\begin{itemize}
  \item 
  At each discrete time step, referred to as a prediction update point, we use a neural network to forecast the behavior of  the uncontrollable agent over the next control period $\Delta t$. 
  \item 
  We apply conformal prediction techniques to expand the predicted results into a confidence region that contains the agent’s reachable states with high probability.
  \item 
  Then within each control period $\Delta t$, we apply the technique of control barrier functions to generate control inputs continuously to avoid reaching the predicted region of the uncontrollable agent. 
  \item 
  The above process is repeated for prediction update point indefinitely.
\end{itemize}
To implement the generic idea of the above framework, the main challenge lies in the coupling between the controlled system and the uncontrollable agent. Specifically, the behavior of the uncontrollable agent depends on the behavior of the controlled system. Therefore, one cannot simply design a predictor and apply conformal prediction solely based on the trajectory data of the uncontrollable agent, as the agent will behave differently when the controlled system is in a different state. This issue is the well-known \emph{distribution shift} problem in conformal prediction. To address this challenge, our approach is to design a predictor and perform conformal prediction over an \emph{augmented state space} that accounts for the joint behavior of both the controlled system and the uncontrollable agents.

\subsection{Data Collection and Predictor Design}
To collect   data and to predict the behavior of the uncontrollable agent, we sample the trajectories of both the controlled system and the uncontrollable agent with time interval $\Delta t$.  
We denote by $[X_0, X_1, \ldots, X_k]$ and $[Y_0, Y_1, \ldots, Y_k]$ the sampled trajectories. For convenience, let $X$ denote the state space of the controllable agent which includes both position and velocity, while $Y$ represents the pedestrian's position.
To capture the coupling between the controlled system and the uncontrollable agents,  
at each sampling instant, the  movement of the uncontrollable agent, i.e., 
$\Delta Y_{i}:=  Y_{i+1}-Y_i $, follows a distribution parameterized by the state of the controlled system in \eqref{eq:uncontrol_dynamic}, i.e., 
\begin{equation}
    \Delta Y_{i}  \sim \mathcal{D}(X_i,Y_i) 
    \label{eq:uncontrollable_distribution}. 
\end{equation} 
Note that distribution $\mathcal{D}(X_i, Y_i)$ is memoryless and time-invariant meaning that the behavior of the uncontrollable agent only depends on the current state of the system rather than the entire history. This assumption is valid in many applications such as the behavior of moving pedestrians. 
Furthermore, this distribution is \emph{unknown a priori} and we can only sample data according to the distribution.


\noindent\textbf{Data Collections.} \quad
We assume that we have access to $\bar{K}$ independent trajectories of observations
$X^{(j)} = (X_1^{(j)}, X_2^{(j)}, \ldots)$, $Y^{(j)} = (Y_1^{(j)}, Y_2^{(j)}, \ldots)$. 
Let $Z^{(j)}_{(0,t]} := (Z_1^{(j)}, \dots, Z_t^{(j)})$ where $Z_{i}^{(j)} = (X_i^{(j)},Y_i^{(j)})$. 

Let $\Delta Y^{(j)} = (\Delta Y_1^{(j)}, \Delta Y_2^{(j)}, \ldots) , \forall j = 1,2,\ldots, \bar{K}$ sampled from the random distribution of $(X^{(j)},Y^{(j)})$, we denote $W^{(j)} = [(Z^{(j)}_{(0,1]},\Delta Y^{(j)}_1),(Z^{(j)}_{(0,2]},\Delta Y^{(j)}_2),\dots]$ 

We partition these trajectories into calibration and training sets, denoted as $D_{\text{cal}} := \{W^{(1)}, \dots, W^{(|D_{\text{cal}}|)}\}$ and $D_{\text{train}} := \{W^{(|D_{\text{cal}}|+1)}, \dots, W^{(\bar{K})}\}$, respectively.

\noindent\textbf{Trajectory Predictor.} \quad 
Our objective is to use the training set $D_{\text{train}} $ to design a predictor to forecast the future state of the uncontrollable agent. 
Note that, since $\Delta Y_{i}+Y_i = Y_{i+1} $, it suffices to predict $\Delta Y_{i}$ as we assume that all state information can be observed perfectly. 
In principle, although the $\Delta Y_{i}$ only depends on the current states, we still use the entire past samples to perform the prediction to improve the prediction accuracy considering the random disturbance influence on the uncontrollable agent dynamic. 
Therefore, the predictor is of form 
\begin{equation}
    \Delta\hat{Y}_{t} = \Omega(Z_{(0,t]}|\theta) \label{eq:NN}.
\end{equation}  
Note that $\Omega$ could be implemented by any trajectory predictor such as neural networks (NNs) with parameters $\theta$.

\subsection{Uncertainty Qualification with Conformal Prediction}
The prediction $\Delta\hat{Y}_{t}$ generated by the trained model, is not guaranteed to be accurate. To address this, we employ conformal prediction to construct a prediction region for $\Delta\hat{Y}_{t}$ that holds with high probability.



\begin{mydef}
For a joint trajectory $Z^{(j)}_{(0,i]} = (Z_{1}^{(j)}, Z_{2}^{(j)}, \ldots, Z_{i}^{j})$ and $W^{(j)}_{i} = (Z^{(j)}_{(0,i]}, \Delta Y^{(j)}_{i})\in D_{\text{cal}}$, the nonconformity score $R^{(j)}$ is defined as the prediction error
\begin{equation}
    R^{(j)} := \sup _{W^{(j)}_{i} \in W^{(j)}}||\Delta Y_{i}^{(j)} - \Delta \hat{Y}_{i}^{(j)}||,
\end{equation} 
where $\Delta Y_{i}^{(j)}$ is the ground truth of a predicted $\Delta \hat{Y}_{i}^{(j)}$.
\end{mydef}
The following result states how to use the  calibration dataset to generate a confidence regions that account for potential variations in the agent's behavior. 
\begin{myprop}
Given a calibration dataset $D_{\text{cal}} $, and a failure probability tolerance $\alpha \in [0,1]$, by \cite{tibshirani2019conformal}, $\forall X \in \mathbb{R}^{n_X}$, $t \in \mathbb{R}_{\geq 0}$, $\operatorname{Prob}\left(R^{(0)} \leq \bar{R}(\alpha) \right) \geq 1-\alpha$ holds where 
\begin{equation}
    \bar{R}(\alpha) = \text { Quantile }_{1-\alpha}\left(R^{(1)}, \ldots, R^{(|D_{\text{cal}}|)}, \infty\right)
    \label{eq:CPbound}
\end{equation}
Note that here $\bar{R}(\alpha) \in \mathbb{R}^m$ is a vector for each dimension and $|D_{\text{cal}}|$ is the calibration dataset size. Similarly, we obtain $\bar{R} = R^{(r)}$, where $r = \lceil |D_{\text{cal}}| \cdot (1-\alpha) \rceil$. As before, for meaningful prediction regions, we require $|D_{\text{cal}}| \geq \lceil (|D_{\text{cal}}|+1)(1-\alpha) \rceil$, otherwise $\bar{R}(\alpha) = \infty$.
\end{myprop}

Conformal prediction implies the nonconformity score $R$ can essentially be bounded, with probability $1-\alpha$, within the conformal region $\bar{R}(\alpha)$. 
Given that the test and calibration data are exchangeable, we can hence conclude that for each coordinate direction $i = 1,...,m,$
\begin{equation}
     \operatorname{Prob}\left[\left( \frac{| y_i - (Y_t + \Delta \hat{Y}_{t})_i |}{\bar{R}_i(\alpha)} \right) \leq 1\right]  \geq 1-\alpha \label{eq:CP_bound_final}
\end{equation}
From known measurement $X_t$ and $Y_t$, and prediction $\Delta \hat{Y}_{t}$, we denote the conformal region for $Y_{t+1}$ in a probabilistic sense as 


\begin{equation}
\hat{\mathcal{O}}_{Y_{t+1}}(\alpha) = \left\{ y \in \mathbb{R}^m \mid \max_{i} \left( \frac{| y_i - (Y_t + \Delta \hat{Y}_{t})_i |}{\bar{R}_i(\alpha)} \right) \leq 1 \right\}
\end{equation}
This set represents a \emph{weighted $\ell_{\infty}$-norm ball} centered at $(Y_t + \Delta \hat{Y}_{t}) $, where $\bar{R}(\alpha) $ determines the scaling in each coordinate direction. It satisfies that 
\begin{equation}
    \operatorname{Prob}[Y_{t+1} \in \hat{\mathcal{O}}_{ Y_{t+1}}(\alpha)]  \geq 1-\alpha.
    \label{eq:safe_region}
\end{equation}
Note that $ \hat{\mathcal{O}}_{ Y_{t+1}}(\alpha)$ is the estimated conformal prediction region of the uncontrollable agent in the next sampling step. 
\subsection{Safe Control using Control Barrier Functions}

After we quantified the uncertainty with the calculated prediction region, we design an uncertainty-aware safety filter with CBF.

\begin{mydef}(\cite{ames2016control})
Let $\mathcal{C} \subseteq \mathbb{R}^{n}$ be the zero-superlevel set of a continuously differentiable function $h :\mathbb{R}^{n} \rightarrow \mathbb{R}$. The function $h$ is a control barrier function (CBF) for the system in Eq.~\eqref{eq:control_dynamic} if there exists an extended class $\mathcal{K}_\infty$ function $\beta$ such that
\begin{equation}
    \sup_{u\in\mathcal{U}} [L_fh(x) + L_gh(x)] \geq -\beta(h(x)) \label{eq:CBF}
\end{equation}
where $L_fh(x)$ and $L_gh(x)$ represent the Lie derivatives.  
\end{mydef}

At each time step, the state of the uncontrollable agent is treated as a fixed parameter, reducing $h_{xy}(x,y)$ to $h(x)$. Therefore, Problem~1 can be formulated as a safe control problem within each prediction update period by considering the estimated region of the uncontrollable agent as an obstacle simply by define 
\begin{equation}
\mathcal{C}=
\{  x\in \mathbb{R}^n\mid \forall y\in \mathbb{R}^{m}:(x,y)\in \mathcal{S}\}
\end{equation}
as the safe set for the controlled system such that the overall system remain in safe set $\mathcal{S}$. 
Therefore, it suffices to find a CBF $h(x)$ 
such that 
\begin{equation}
C=\{x\in \mathbb{R}^n\mid {h}(x)\geq 0 \}\subseteq \mathcal{C}.
\end{equation}
In the literature, there are many existing approach for finding such CBF $h(x)$, e.g., by sum-of-square programming \cite{zhao2023convex} or learning-based approach \cite{robey2020learning,srinivasan2020synthesis}. 

 Now, given the synthesized CBF $h(x)$, the set of admissible control inputs is defined by 
\begin{equation}
    K_{CBF} := \{u \in \mathcal{U}\mid L_f h(x) + L_g h(x)u 
     \geq -\beta(h(x))\}
\label{eq:KCBF}
\end{equation}
The safe input can be any control law $u(t)$ from $K_{CBF}$ in Eq.~\eqref{eq:KCBF}, e.g., one computed by solving the quadratic optimization problem \eqref{eq:optimal_goal}.

\begin{myprop}
Consider the control law $u(t)$ from $K_{CBF}$ in Eq.~\eqref{eq:KCBF},the system composed by controllable and uncontrollable variables Eq.~\eqref{eq:control_dynamic}, Eq.~\eqref{eq:uncontrollable_distribution}  satisfies Eq.~\eqref{eq:CBF}, 
respectively, 
we have
\begin{equation}
    \operatorname{Prob}[(X_{t+1},Y_{t+1})\in \mathcal{S}] \geq 1-\alpha \label{eq:finaleq}
\end{equation} 
\end{myprop}
\begin{proof}
With guarantee that $\operatorname{Prob}[Y_{t+1}\in \hat{\mathcal{O}}_{Y_{t+1}}(\alpha)] \geq 1-\alpha$, the control law $u(t)$ derived from $K_{CBF}$ for $\hat{\mathcal{C}}$ is effective in ensuring the safety of the entire system, meaning that $(X_{t+1}, Y_{t+1})$ remains within the safe set $\mathcal{S}$, as defined in Problem~1.
\end{proof}

\section{Simulation Results}

In this section, we illustrate and evaluate the proposed approach through simulations of a case study involving autonomous driving with uncontrollable pedestrians. This case study demonstrates the effectiveness of our framework in modelling the interactions between the autonomous vehicle and pedestrians, providing probabilistic safety guarantees under the uncertainty of pedestrian behavior.

\begin{figure}[t]
    \centering
    \includegraphics[width=0.48\textwidth]{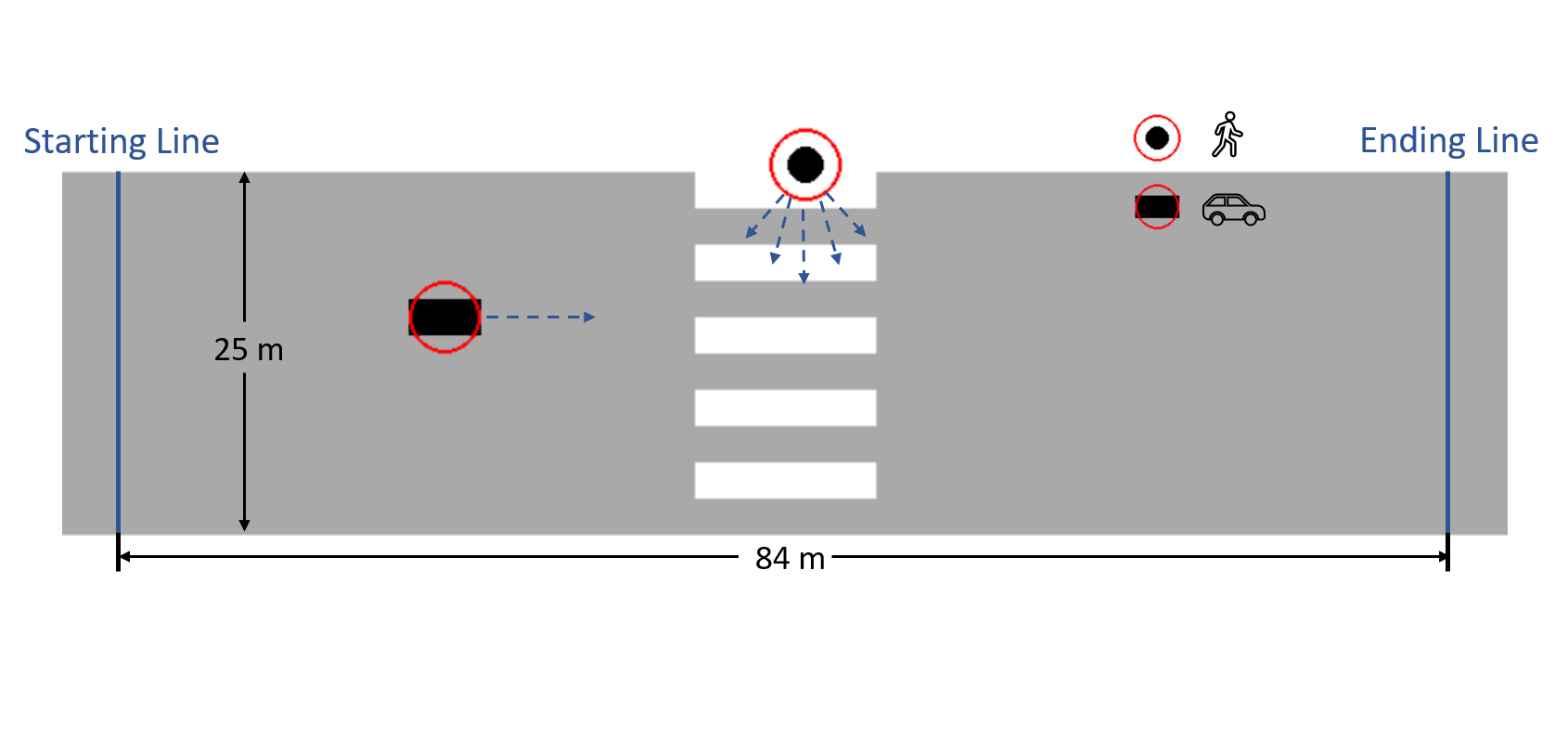} 
    \caption{The pedestrian's speed follows an unknown stochastic distribution affected by the state of the vehicle, i.e.,$\Delta Y \sim \mathcal{D}(X,Y)$. The pedestrian starts at a random position on one side of the crosswalk while the vehicle begins at the midpoint of the starting line. Red circles indicate collision volume.}
    \label{fig:map}
\end{figure}

\begin{figure*}[t]   
    \centering
    \begin{subfigure}[b]{0.4\textwidth}
        \centering
        \includegraphics[width=\textwidth]{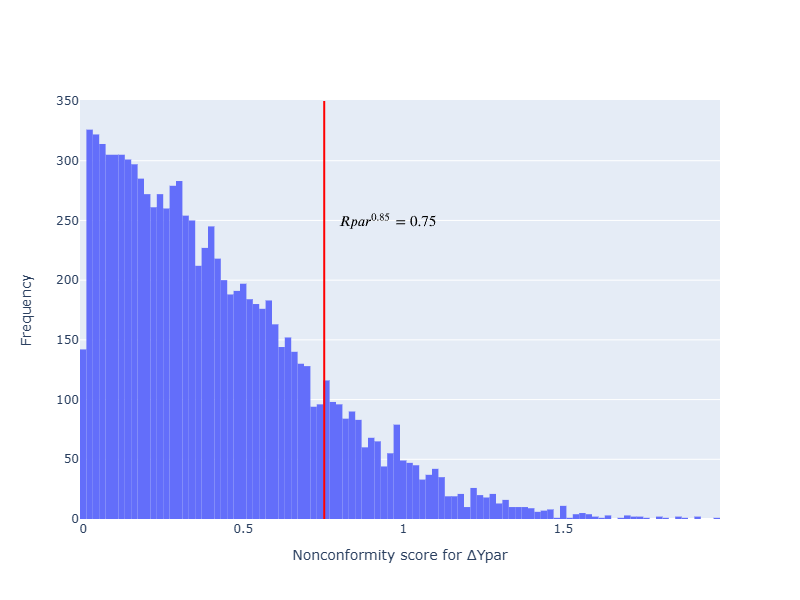} 
    \end{subfigure}
    \begin{subfigure}[b]{0.4\textwidth}
        \centering
        \includegraphics[width=\textwidth]{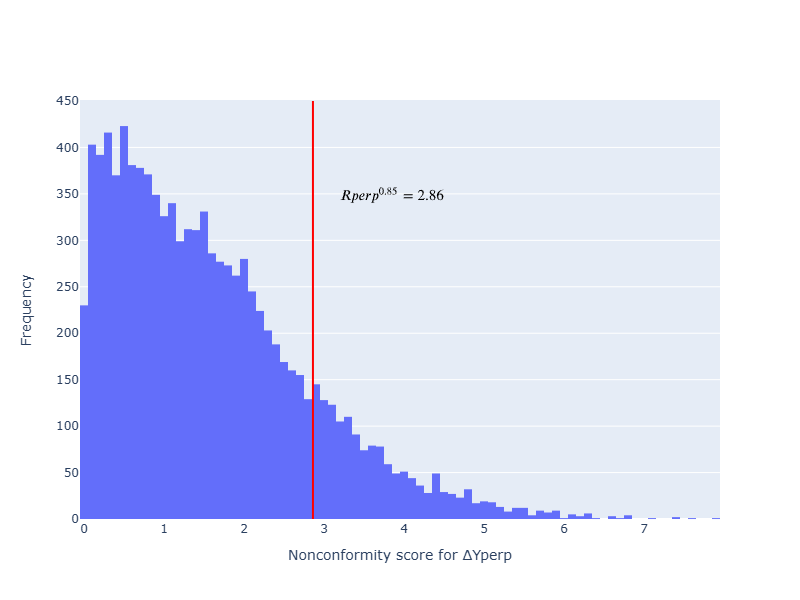} 
    \end{subfigure}
    \caption{Nonconformity scores histogram for pedestrian position displacement $(\Delta {Y_{par}},\Delta {Y_{perp}})$  on $D_{cal}$}
    \label{fig:score}
\end{figure*}

\subsection{Scenario Description and System Dynamics}
\textbf{Scenario Setup:} As illustrated in Fig. \ref{fig:map}, the scenario is set on a straight road segment with a total length of 84 m. A crosswalk is located within this segment, and the road has a width of 25 m. The starting line for the vehicle is positioned at one end of the road, while the ending line is at the opposite end. A pedestrian is present near the crosswalk, and a vehicle is approaching from the starting line. The control period is set to be $\Delta t $ = 0.1s.

\textbf{Vehicle Dynamics:}
For the sake of simplicity, we assume the vehicle moves straight through the crosswalk, setting the dimension \(n = 1\) with dynamic model $\dot{x}=u$. This choice aligns with real-world constraints, as vehicles rarely make left or right turns while crossing a zebra line. Nevertheless, our approach naturally extends to two or higher dimensions.  
The vehicle's speed range is \([0, 15]\)m/s, reflecting typical vehicle speeds observed in urban environments (up to the speed limit of 15 m·s$^{-1}$/54 km·h$^{-1}$/33 mph). The vehicle's position is updated based on its velocity.

\textbf{Pedestrian Behavior:}
The pedestrian's movement is modeled as a \emph{stochastic process} governed by a function of the vehicle's state-space vector and Gaussian noise. Mathematically, this can be expressed as:
\begin{equation}\label{unknown_distribution}
    \Delta{Y} \sim \mathcal{N}_2(\phi(X, Y),\Sigma), 
\end{equation}
where $\Delta Y  = (\Delta  {Y_{par}}, \Delta  {Y_{perp}}) \in \mathbb{R}^2$ follows a Gaussian distribution whose mean value is a function of both the state of the pedestrian and the controlled vehicle. Here the notation $\Delta  {Y_{par}}$ represents the velocity component parallel to the road direction, while $\Delta  {Y_{perp}}$ denotes perpendicular component. $\Sigma$ is the standard deviation $\Sigma = \operatorname{diag}(\sigma^2_x,\sigma^2_y) = \operatorname{diag}((\text{0.5m/s})^2, (\text{2m/s})^2)$. 
To illustrate, as the vehicle approaches at a higher speed, the pedestrian slows down $\Delta  {Y_{perp}}$, exhibiting more cautious crossing behavior. Simultaneously, $\Delta  {Y_{par}}$ increases in the direction of moving away from the vehicle as an evasive response.
Additionally, $\Delta  {Y_{par}}$ is influenced by $Y$. When approaching the left/right edge of the crosswalk area, the pedestrian tends to adjust the position toward the center of the crosswalk.
And as pedestrian movement is inherently unpredictable, this modelling introduces additional stochasticity to account for hesitation and spontaneous decision-making. 

Under this simulation setup, pedestrians exhibit a perturbed yet generally arced crossing trajectory around the vehicle. This behavior aligns with the natural tendency of pedestrians to adjust their paths dynamically while crossing. 

\subsection{Model Implementation}
\textbf{Interaction and Collision Avoidance:}
The pedestrian's trajectory is evaluated in real-time as the vehicle approaches. 
We track the pedestrian's crossing behavior which is influenced by the vehicle. The simulation continuously monitors the distance between the vehicle and the pedestrian, and the safe controller (SPARC) is applied to ensure safe control synthesis in avoiding collisions.

\textbf{Trajectory Predictor:} We trained a feedforward neural network to predict the future trajectory of the pedestrian. 
The dataset is generated by simulating a vehicle moving at a constant speed while tracking both the vehicle and pedestrian trajectories. The vehicle's speed is uniformly sampled from the range \([0, 15]\)m/s.
The collected data is structured following the format described in the \textit{IV. B} section.
The input is the history state measurement of the vehicle's and pedestrian's state. The training dataset $D_{\text {train }}$ contains $1 \times 10^6$ data, and the calibration dataset $D_{\text {cal}}$ for conformal prediction contains $1 \times 10^5$ data. 
The model consists of three fully connected layers with ReLU activation.
The input layer maps features to 64 hidden units, followed by a second layer of 64 units, and a final output layer predicting pedestrian speed $(\Delta  {Y_{par}}, \Delta  {Y_{perp}}) \in \mathbb{R}^2 $. The model use the Adam optimizer (learning rate 0.001) and MSE loss. A StepLR scheduler halves the learning rate every 10 epochs over 50 epochs.

\textbf{Control Barrier Functions:} For vehicle-pedestrian collision avoidance, we choose the following CBF:

\begin{equation}
    h(x) = \min_{y \in \hat{\mathcal{O}}_{y}} ||x - y|| 
     - d_\text{safe}
\end{equation}
where:
\begin{itemize}
    \item $x \in \mathbb{R}^{n}$ is the position of the vehicle,
    \item $\hat{\mathcal{O}}_{y} \subseteq \mathbb{R}^{m}$ is the predicted conformal region of the uncontrollable pedestrian, and we set $\hat{\mathcal{O}}_{y}(\alpha) = [\hat{y} - \bar{R}(\alpha), \hat{y} + \bar{R}(\alpha)] \subseteq \mathbb{R}^2$ as a weighted $\ell_\infty$-norm ball.
    \item $d_\text{safe}$ is the minimum safe distance to be maintained,
    \item $\|\cdot\|$ denotes the Euclidean norm.  
\end{itemize}

With $d_{\text{safe}}$ = 1.0 m, the collision volume is shown as the red circle in Fig. \ref{fig:map}.

To demonstrate the effectiveness of our method, we compare the following two cases in simulations:
\begin{enumerate}
    \item SPARC: Let $\alpha$ be the failure tolerance, randomly pick $\alpha = 0.85, 0.75, 0.5$, the results are shown in Table \ref{tab:collprob}.
    \item Vanilla CBF: We set $\alpha = 1$, which means that $\bar{R}(\alpha)=[0, 0]$.
    This  essentially reduces to the vanilla CBF \cite{ames2016control}. This method uses the direct output from the neural network output than the calibrated conformal region, providing no safety guarantee and failing to account for dynamic stochasticity.
    \item Random: In this experiment, the vehicle's initial velocity is uniformly sampled, and it follows the reference control $u_r$ during its motion. The proportion of trials resulting in a collision with the pedestrian is recorded. 
    
\end{enumerate}

\begin{table}[t]
\caption{Collision Probability}
\begin{center}
\begin{tabular}{c|c|c}
\hline
Controller & Collision(Absolute) & Collision(Relative)  \\ \hline
SPARC ($\alpha$=0.15)    & 1.62\%               & 14.46\% \textless{} 15\%         \\ \hline
SPARC ($\alpha$=0.25)     & 1.96\%                         & 17.50\% \textless{} 25\% \\ \hline
SPARC ($\alpha$=0.50)     & 4.21\%                        & 37.59\% \textless{} 50\%         \\ \hline
Vanilla CBF     & 10.80\%                         & 96.43\%  \\ \hline
Random     & 11.20\%                         & 100.00\%  \\ \hline
\end{tabular}
\end{center}
\label{tab:collprob}
\end{table}

\begin{figure}[thpb]
  \centering
  \includegraphics[width=0.48\textwidth]{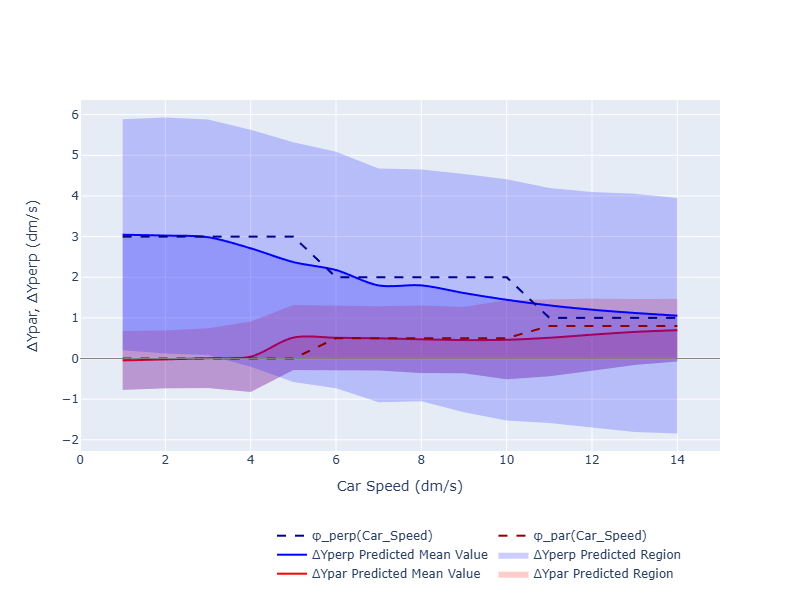}
  \caption{The prediction output with conformal prediction region. The plot also includes $\phi(\text{car speed})$ with all other parameters fixed as a noise-free reference ground truth.}
  \label{fig:cp_region}
\end{figure}
\subsection{Result Analysis}
\textbf{Trajectory Predictors and Conformal Prediction:}
Take failure probability tolerance $\alpha = 0.85$ as an example, Fig. \ref{fig:score} illustrates a calculation of the nonconformity scores used in our approach. 
The scores $R_ {Y_{par}}^{0.85}=0.75, R_ {Y_{perp}}^{0.85}=2.86$ represent the deviation between the predicted pedestrian position change $\Delta\hat{Y}$ and the ground truth, providing a measure of prediction uncertainty
, which is used to generate safe control decisions.

Fig. \ref{fig:cp_region} shows the how the network learns the underlying random distribution governing pedestrian behavior in relation to the vehicle’s state. Notably, our method does not rely on the precise accuracy, overall performance, or scene comprehension capabilities of the model, making it robust to model imperfections and uncertainties.
In our simulation framework, for the sake of simplicity, the upper limit of all nonconformity scores is employed as a proxy for the overall safety margin, and it is also feasible to utilize a grid-based distribution of nonconformity scores to devise a more aggressive control strategy.

\textbf{Safe Control Barrier Function Controller:}
The vehicle controlled by our SPARC framework is expected to avoid collisions with pedestrians. 
We set the failure probability tolerance $\alpha=0.15, 0.25, 0.5$, and desire $\text{Prob}[(X,Y) \in \mathcal{S}] \geq 1-\alpha$, which means $\text{Prob}[\text{vehicle and pedestrian collide}] \leq \alpha$. We performed $1\times 10^5$ experiments, each time step is set to 100, the collision probability are collected in Table \ref{tab:collprob}.
In 11.20\% of the trials (1120 experiments), the car crash on the pedestrian, and we care about whether our controller can successfully prevent collision in these experiment, so we divide absolute collision frequency by 11.20\% to get relative collision frequency. 
Due to the relatively high ratio of the standard deviation of the added stochastic component to the pedestrian's average speed, the vanilla CBF, which relies solely on predicted pedestrian positions for avoidance, struggles to account for this uncertainty and consequently suffers from a high failure rate$(96.43 \%)$. The collision rate of SPARC are $14.46\%(<15\%), 17.50\%(<25\%), 37.59\%(<50\%)$, which are all under their respective failure probability tolerance and significantly lower than vanilla CBF case $(96.43 \%)$, so our Proposition 2. holds in practice.

\section{Conclusion}
This work presents a new framework that integrates trajectory prediction with conformal prediction to enhance safety in multi-agent systems involving coupled controllable and uncontrollable agents. By incorporating probabilistic guarantees into the control barrier function framework, we ensure safe interactions between the controlled system and uncontrollable agents with high confidence.
Compared to previous approaches, our method accounts for the coupling between the controlled system and the uncontrollable agent by sampling and evaluating over an augmented state space, making it more applicable to real-world scenarios such as autonomous driving and human-robot interaction.

There are several future directions for this work that we plan to explore. 
The main objective of our work is to present a generic framework applicable to a broad class of control synthesis problems involving coupled uncontrollable agents, rather than being tailored to specific applications. In this work, for illustrative purposes, we use simplified vehicle and pedestrian models to demonstrate our approach.   
Regarding the complex nature of vehicle-pedestrian interactions, various datasets and prediction models have been developed to capture realistic behaviors. Notable datasets such as ETH \& UCY \cite{pellegrini2009you} and JAAD \cite{rasouli2017JAAD} provide benchmarks for pedestrian trajectory modeling, while methods like Social-LSTM \cite{alahi2016social} and Trajectron++ \cite{salzmann2020trajectron++} offer deep learning-based approaches for trajectory prediction. The TrajNet++ benchmark \cite{kothari2021human} further standardizes evaluation in this domain. In the future, we plan to focus specifically on autonomous driving by leveraging these real-world datasets and customized prediction methods to enhance the accuracy and applicability of our framework.  
Additionally, one limitation of our current approach is the assumption that precise state information for both the controlled system and the uncontrollable agent is fully accessible. In many practical applications, such state information must be obtained through perception devices like LiDAR, which can introduce estimation errors. To address this, we plan to extend our framework to perception-based settings, incorporating techniques to handle state estimation uncertainty and further improve the robustness of our approach in real-world scenarios.

\bibliographystyle{IEEEtran}
\bibliography{ref}

\begin{thebibliography}{10}
\providecommand{\url}[1]{#1}
\csname url@samestyle\endcsname
\providecommand{\newblock}{\relax}
\providecommand{\bibinfo}[2]{#2}
\providecommand{\BIBentrySTDinterwordspacing}{\spaceskip=0pt\relax}
\providecommand{\BIBentryALTinterwordstretchfactor}{4}
\providecommand{\BIBentryALTinterwordspacing}{\spaceskip=\fontdimen2\font plus
\BIBentryALTinterwordstretchfactor\fontdimen3\font minus \fontdimen4\font\relax}
\providecommand{\BIBforeignlanguage}[2]{{%
\expandafter\ifx\csname l@#1\endcsname\relax
\typeout{** WARNING: IEEEtran.bst: No hyphenation pattern has been}%
\typeout{** loaded for the language `#1'. Using the pattern for}%
\typeout{** the default language instead.}%
\else
\language=\csname l@#1\endcsname
\fi
#2}}
\providecommand{\BIBdecl}{\relax}
\BIBdecl

\bibitem{inD}
J.~Bock, R.~Krajewski, T.~Moers, S.~Runde, L.~Vater, and L.~Eckstein, ``The ind dataset: A drone dataset of naturalistic road user trajectories at german intersections,'' in \emph{2020 IEEE Intelligent Vehicles Symposium (IV)}.\hskip 1em plus 0.5em minus 0.4em\relax IEEE, 2020, pp. 1929--1934.

\bibitem{SDD}
A.~Robicquet, A.~Sadeghian, A.~Alahi, and S.~Savarese, ``Learning social etiquette: Human trajectory understanding in crowded scenes,'' in \emph{ECCV}.\hskip 1em plus 0.5em minus 0.4em\relax Springer, 2016, pp. 549--565.

\bibitem{yang2022differentiable}
S.~Yang, S.~Chen, V.~M. Preciado, and R.~Mangharam, ``Differentiable safe controller design through control barrier functions,'' \emph{IEEE Control Systems Letters}, vol.~7, pp. 1207--1212, 2022.

\bibitem{xiao2021high}
W.~Xiao and C.~Belta, ``High-order control barrier functions,'' \emph{IEEE Transactions on Automatic Control}, vol.~67, no.~7, pp. 3655--3662, 2021.

\bibitem{xiao2021adaptive}
W.~Xiao, C.~Belta, and C.~G. Cassandras, ``Adaptive control barrier functions,'' \emph{IEEE Transactions on Automatic Control}, vol.~67, no.~5, pp. 2267--2281, 2021.

\bibitem{xiaosafe}
W.~Xiao, C.~G. Cassandras, and C.~Belta, ``Safe autonomy with control barrier functions.''

\bibitem{ames2019control}
A.~D. Ames, S.~Coogan, M.~Egerstedt, G.~Notomista, K.~Sreenath, and P.~Tabuada, ``Control barrier functions: Theory and applications,'' in \emph{2019 18th European control conference (ECC)}.\hskip 1em plus 0.5em minus 0.4em\relax IEEE, 2019, pp. 3420--3431.

\bibitem{ames2016control}
A.~D. Ames, X.~Xu, J.~W. Grizzle, and P.~Tabuada, ``Control barrier function based quadratic programs for safety critical systems,'' \emph{IEEE Transactions on Automatic Control}, vol.~62, no.~8, pp. 3861--3876, 2016.

\bibitem{jankovic2018robust}
M.~Jankovic, ``Robust control barrier functions for constrained stabilization of nonlinear systems,'' \emph{Automatica}, vol.~96, pp. 359--367, 2018.

\bibitem{buch2021robust}
J.~Buch, S.-C. Liao, and P.~Seiler, ``Robust control barrier functions with sector-bounded uncertainties,'' \emph{IEEE Control Systems Letters}, vol.~6, pp. 1994--1999, 2021.

\bibitem{nejati2023data}
A.~Nejati and M.~Zamani, ``Data-driven synthesis of safety controllers via multiple control barrier certificates,'' \emph{IEEE Control Systems Letters}, vol.~7, pp. 2497--2502, 2023.

\bibitem{zhong2022synthesizing}
B.~Zhong, M.~Zamani, and M.~Caccamo, ``Synthesizing safety controllers for uncertain linear systems: A direct data-driven approach,'' in \emph{2022 IEEE Conference on Control Technology and Applications (CCTA)}.\hskip 1em plus 0.5em minus 0.4em\relax IEEE, 2022, pp. 1278--1284.

\bibitem{ajeleye2024data}
D.~Ajeleye and M.~Zamani, ``Data-driven controller synthesis via co-b{\"u}chi barrier certificates with formal guarantees,'' \emph{IEEE Control Systems Letters}, 2024.

\bibitem{mitsioni2023safe}
I.~Mitsioni, P.~Tajvar, D.~Kragic, J.~Tumova, and C.~Pek, ``Safe data-driven model predictive control of systems with complex dynamics,'' \emph{IEEE Transactions on Robotics}, vol.~39, no.~4, pp. 3242--3258, 2023.

\bibitem{salamati2024data}
A.~Salamati, A.~Lavaei, S.~Soudjani, and M.~Zamani, ``Data-driven verification and synthesis of stochastic systems via barrier certificates,'' \emph{Automatica}, vol. 159, p. 111323, 2024.

\bibitem{chen2023data}
Y.~Chen, C.~Shang, X.~Huang, and X.~Yin, ``Data-driven safe controller synthesis for deterministic systems: A posteriori method with validation tests,'' in \emph{2023 62nd IEEE Conference on Decision and Control (CDC)}.\hskip 1em plus 0.5em minus 0.4em\relax IEEE, 2023, pp. 7988--7993.

\bibitem{yin2024formal}
X.~Yin, B.~Gao, and X.~Yu, ``Formal synthesis of controllers for safety-critical autonomous systems: Developments and challenges,'' \emph{Annual Reviews in Control}, vol.~57, p. 100940, 2024.

\bibitem{shafer2008tutorial}
G.~Shafer and V.~Vovk, ``A tutorial on conformal prediction.'' \emph{Journal of Machine Learning Research}, vol.~9, no.~3, 2008.

\bibitem{lindemann2023safe}
L.~Lindemann, M.~Cleaveland, G.~Shim, and G.~J. Pappas, ``Safe planning in dynamic environments using conformal prediction,'' \emph{IEEE Robotics and Automation Letters}, 2023.

\bibitem{yang2023safe}
S.~Yang, G.~J. Pappas, R.~Mangharam, and L.~Lindemann, ``Safe perception-based control under stochastic sensor uncertainty using conformal prediction,'' in \emph{2023 62nd IEEE Conference on Decision and Control (CDC)}.\hskip 1em plus 0.5em minus 0.4em\relax IEEE, 2023, pp. 6072--6078.

\bibitem{dixit2023adaptive}
A.~Dixit, L.~Lindemann, S.~X. Wei, M.~Cleaveland, G.~J. Pappas, and J.~W. Burdick, ``Adaptive conformal prediction for motion planning among dynamic agents,'' in \emph{Learning for Dynamics and Control Conference}.\hskip 1em plus 0.5em minus 0.4em\relax PMLR, 2023, pp. 300--314.

\bibitem{yu2023}
\BIBentryALTinterwordspacing
X.~Yu, Y.~Zhao, X.~Yin, and L.~Lindemann, ``Signal temporal logic control synthesis among uncontrollable dynamic agents with conformal prediction,'' 2023. [Online]. Available: \url{https://arxiv.org/abs/2312.04242}
\BIBentrySTDinterwordspacing

\bibitem{stamouli2024recursively}
C.~Stamouli, L.~Lindemann, and G.~Pappas, ``Recursively feasible shrinking-horizon mpc in dynamic environments with conformal prediction guarantees,'' in \emph{6th Annual Learning for Dynamics \& Control Conference}.\hskip 1em plus 0.5em minus 0.4em\relax PMLR, 2024, pp. 1330--1342.

\bibitem{sun2024conformal}
J.~Sun, Y.~Jiang, J.~Qiu, P.~Nobel, M.~J. Kochenderfer, and M.~Schwager, ``Conformal prediction for uncertainty-aware planning with diffusion dynamics model,'' \emph{Advances in Neural Information Processing Systems}, vol.~36, 2024.

\bibitem{tibshirani2019conformal}
R.~J. Tibshirani, R.~Foygel~Barber, E.~Candes, and A.~Ramdas, ``Conformal prediction under covariate shift,'' \emph{Advances in neural information processing systems}, vol.~32, 2019.

\bibitem{zhao2023convex}
P.~Zhao, R.~Ghabcheloo, Y.~Cheng, H.~Abdi, and N.~Hovakimyan, ``Convex synthesis of control barrier functions under input constraints,'' \emph{IEEE Control Systems Letters}, 2023.

\bibitem{robey2020learning}
A.~Robey, H.~Hu, L.~Lindemann, H.~Zhang, D.~V. Dimarogonas, S.~Tu, and N.~Matni, ``Learning control barrier functions from expert demonstrations,'' in \emph{59th IEEE Conference on Decision and Control (CDC)}.\hskip 1em plus 0.5em minus 0.4em\relax IEEE, 2020, pp. 3717--3724.

\bibitem{srinivasan2020synthesis}
M.~Srinivasan, A.~Dabholkar, S.~Coogan, and P.~A. Vela, ``Synthesis of control barrier functions using a supervised machine learning approach,'' in \emph{IEEE/RSJ International Conference on Intelligent Robots and Systems (IROS)}.\hskip 1em plus 0.5em minus 0.4em\relax IEEE, 2020, pp. 7139--7145.

\bibitem{pellegrini2009you}
S.~Pellegrini, A.~Ess, K.~Schindler, and L.~Van~Gool, ``You'll never walk alone: Modeling social behavior for multi-target tracking,'' in \emph{IEEE 12th International Conference on Computer Vision}.\hskip 1em plus 0.5em minus 0.4em\relax IEEE, 2009, pp. 261--268.

\bibitem{rasouli2017JAAD}
A.~Rasouli, I.~Kotseruba, T.~Kunic, and J.~K. Tsotsos, ``Jaad: A dataset for studying pedestrian behavior in the context of autonomous driving,'' in \emph{Proceedings of the IEEE International Conference on Computer Vision Workshops (ICCVW)}.\hskip 1em plus 0.5em minus 0.4em\relax IEEE, 2017, pp. 1919--1927.

\bibitem{alahi2016social}
A.~Alahi, K.~Goel, V.~Ramanathan, A.~Robicquet, L.~Fei-Fei, and S.~Savarese, ``Social lstm: Human trajectory prediction in crowded spaces,'' in \emph{Proceedings of the IEEE Conference on Computer Vision and Pattern Recognition (CVPR)}.\hskip 1em plus 0.5em minus 0.4em\relax IEEE, 2016, pp. 961--971.

\bibitem{salzmann2020trajectron++}
T.~Salzmann, B.~Ivanovic, P.~Chakravarty, and M.~Pavone, ``Trajectron++: Dynamically-feasible trajectory forecasting with heterogeneous data,'' in \emph{Proceedings of the European Conference on Computer Vision (ECCV)}.\hskip 1em plus 0.5em minus 0.4em\relax Springer, 2020, pp. 683--700.

\bibitem{kothari2021human}
P.~Kothari, S.~Kreiss, and A.~Alahi, ``Human trajectory forecasting in crowds: A deep learning perspective,'' \emph{IEEE Transactions on Intelligent Transportation Systems}, vol.~23, no.~7, pp. 7384--7396, 2021.

\end{thebibliography}

\end{document}